\documentclass[11pt,letterpaper]{article}
\usepackage[letterpaper, left=1in, right=1in, top=0.9in, bottom=0.9in]{geometry}
\usepackage[american]{babel}
\usepackage[normalem]{ulem}
\usepackage{amsmath, amssymb, cases, amsthm}
\usepackage{thmtools}
\usepackage[shortlabels]{enumitem}
\usepackage{mdframed}
\usepackage{bbm}
\usepackage{bm}
\usepackage{microtype}
\usepackage{xcolor}
\usepackage{makecell}
\usepackage{mathtools}
\usepackage{algorithmic}
\usepackage[procnumbered,ruled,vlined,linesnumbered]{algorithm2e}
\usepackage{float}
\usepackage{varwidth}
\usepackage{modletter}
\usepackage{tcolorbox}
\usepackage{mathrsfs}
\newtcolorbox{construction}[2][]
{
	colframe = gray!50,
	colback  = gray!10,
	coltitle = gray!10!black,
	left*=0mm, 
	before skip = 10pt,
	after skip = 10pt,
	title    = \textbf{\space\space #2},
	#1,
}

\SetKwInput{KwData}{Input}
\SetKwInput{KwResult}{Output}
\SetKwInput{KwGlobalVar}{Global variables}

\usepackage[bookmarks,colorlinks,breaklinks]{hyperref}
\hypersetup{urlcolor=blue, colorlinks=true, citecolor=green!50!black, linkcolor=blue}
\usepackage[capitalize,nosort,nameinlink]{cleveref}

\declaretheorem[numberwithin=section,refname={Theorem,Theorems},Refname={Theorem,Theorems}]{theorem}

\declaretheorem[numberlike=theorem]{lemma}
\declaretheorem[numberlike=theorem]{proposition}
\declaretheorem[numberlike=theorem]{corollary}

\theoremstyle{definition}

\def\final{0}  %
\ifnum\final=0  %
\newcommand{\todo}[1]{{\color{red}[{\tiny TODO: \bf #1}]\marginpar{\color{red}*}}}
\newcommand{\yonggang}[1]{{\color{blue}[{\tiny Yonggang: \bf #1}]\marginpar{\color{blue}*}}}
\else %
\newcommand{\yonggang}[1]{}
\newcommand{\todo}[1]{}
\fi

\newcommand{\tout}{\text{out}}
\newcommand{\tin}{\text{in}}

\begin{document} 
\sloppy

\title{
Directed and Undirected Vertex Connectivity Problems\\are Equivalent for Dense Graphs
}
\author{
Olivier Fischer\thanks{ETH Zürich, Switzerland, \texttt{olivier.fischer@inf.ethz.ch}}
\and
Yonggang Jiang\thanks{MPI-INF \& Saarland University, \texttt{yjiang@mpi-inf.mpg.de}}
\and
Sagnik Mukhopadhyay\thanks{University of Birmingham, United Kingdom, \texttt{s.mukhopadhyay@bham.ac.uk}}
\and
Sorrachai Yingchareonthawornchai\thanks{Institute for Theoretical Studies, ETH Zürich, Switzerland, \texttt{sorrachai.yingchareonthawornchai@eth-its.ethz.ch}} 
}
\date{}
\maketitle

\begin{abstract}

Vertex connectivity and its variants are among the most fundamental problems in graph theory, with decades of extensive study and numerous algorithmic advances. The directed variants of vertex connectivity are usually solved by manually extending fast algorithms for undirected graphs, which has required considerable effort. In this paper, we present an extremely simple reduction from directed to undirected vertex connectivity for dense graphs.  As immediate corollaries, we vastly simplify the proof for directed vertex connectivity in $n^{2+o(1)}$ time~\cite{LNPSY25}, and obtain a parallel vertex connectivity algorithm for directed graphs with $n^{\omega+o(1)}$ work and $n^{o(1)}$ depth, via the undirected vertex connectivity algorithm of~\cite{BlikstadJMY25}. 

Our reduction further extends to the weighted, all-pairs and Steiner versions of the problem. By combining our reduction with the recent subcubic-time algorithm for undirected weighted vertex cuts~\cite{ChuzhoyT25}, we obtain a subcubic-time algorithm for weighted directed vertex connectivity, improving upon a three-decade-old bound~\cite{HenzingerRG00} for dense graphs. For the all-pairs version, by combining the conditional lower bounds on the all-pairs vertex connectivity problem for directed graphs~\cite{AbboudGIKPTUW19}, we obtain an alternate proof of the conditional lower bound for the all-pairs vertex connectivity problem on undirected graphs, vastly simplifying the proof by \cite{HuangLSW23}.

\end{abstract}

\clearpage
\section{Introduction}

Vertex connectivity and its variants are among the most fundamental graph problems and have received extensive studies over the last several decades \cite{Kleitman1969methods, Podderyugin1973algorithm, EvenT75, Even75, Galil80, EsfahanianH84, Matula87, BeckerDDHKKMNRW82, LinialLW88, CheriyanT91, NagamochiI92, CheriyanR94, Henzinger97, HenzingerRG00, Gabow06, Censor-HillelGK14, NanongkaiSY19,ForsterNYSY20,  SaranurakY22,NSY23,BlikstadJMY25,Korhonen25,ChuzhoyT25,JiangNSY25}.  For an undirected graph, a \emph{vertex cut} is a vertex set whose removal disconnects the graph. The problem of \emph{vertex connectivity} asks to find a minimum vertex cut of a non-complete undirected graph. For directed graphs, a vertex cut is a vertex set whose removal renders the graph not strongly connected. The problem of \emph{directed vertex connectivity} is to find a minimum vertex cut on a non-complete directed graph. 

For a long time, %
progress in solving directed vertex connectivity often involved a manual, ad-hoc process of ``mimicking" fast algorithms from the undirected graphs.  For example, consider the fastest known algorithm %
for vertex connectivity. A long line of research~\cite{NanongkaiSY19,ForsterNYSY20,LNPSY25} culminates in an almost linear-time randomized algorithm for undirected vertex connectivity. For directed graphs, \cite{LNPSY25} extends their techniques and spends 9 more pages to obtain an $n^{2+o(1)}$-time randomized algorithm for directed vertex connectivity, where $n$ denotes the number of vertices.

A similar phenomenon appeared in the algebraic approach, which was first proposed for undirected graphs~\cite{LinialLW88}; later, the same approach was manually adapted to work for directed graphs~\cite{CheriyanT91}. Very recently, this approach has been made faster on undirected graphs using common-neighborhood clustering techniques to obtain an $n^{\omega+o(1)}$-time algorithm~\cite{BlikstadJMY25}, where $\omega \leq 2.371339$ is the matrix multiplication exponent~\cite{AlmanDWXXZ25}. It is conceivable that their techniques can also be extended to directed graphs, but this has not been done yet.

This suggested a deep connection between undirected and directed problems. In this paper, we resolve the missing link by introducing a \textit{simple black-box}  reduction from directed to undirected vertex connectivity for dense graphs. This formalizes the long-observed phenomenon that the problems were related and explains why mimicking was so often successful. 

\begin{theorem}[Informal version of \Cref{thm:reduction}]\label{thm:intro_1}
    The directed vertex connectivity problem on an $n$-vertex directed graph reduces to the undirected vertex connectivity problem on a graph with $2n$ vertices and $O(n^2)$ edges.
\end{theorem}

 The reduction is short and elementary, which is surprising given that vertex connectivity has received considerable attention over the past five decades. %

As a corollary of \Cref{thm:intro_1}, we obtain the following results immediately. Combining with the $m^{1+o(1)}$ undirected vertex connectivity algorithm by \cite{LNPSY25}, we obtain:
\begin{corollary}
There is an $n^{2+o(1)}$-time directed vertex connectivity algorithm.
\end{corollary}
Thus, it is no longer necessary to open the box and mimic the proof of undirected vertex connectivity as done in \cite{LNPSY25}.  Furthermore, since the reduction is fully parallelizable, we also obtain parallel algorithms in this setting. Combining with the $n^{\omega+o(1)}$-time undirected vertex connectivity algorithm by \cite{BlikstadJMY25}, we obtain: 
\begin{corollary}
There is a directed vertex connectivity algorithm with $n^{\omega+o(1)}$-work  and $n^{o(1)}$-parallel depth.
\end{corollary}
Work refers to the total amount of computation across all processors. Parallel depth refers to the length of the longest chain of dependent operations in an algorithm.

\paragraph{All-pairs minimum vertex cuts.}  The simplicity of our reduction facilitates extending the results to many other settings. We believe that our reduction paves the way to showing many problems related to vertex cut are \emph{equivalent} on dense undirected and directed graphs. The reduction simplifies another paper \cite{HuangLSW23}.

\begin{theorem}[Informal version of \Cref{thm:reduction}]\label{thm:intro_3}
    The all-pairs minimum vertex cuts problem on an $n$-vertex directed graph reduces to the same problem on a $2n$-vertex undirected graph with $O(n^2)$ edges.
\end{theorem}

In \cite{AbboudGIKPTUW19}, they proved that all-pair minimum vertex cut has a conditional lower bound of $\widehat \Omega (n^4)$ for combinatorial algorithms on directed graphs. One of the main contributions of \cite{HuangLSW23} was to extend this lower bound to undirected graphs, which takes 7 pages of proofs. Our theorem \Cref{thm:intro_3} immediately implies the same result when combining with \cite{AbboudGIKPTUW19}.%

\paragraph{Weighted vertex mincuts.} The problem of vertex connectivity naturally extends to vertex-weighted graphs. The problem of (weighted) minimum vertex cut is given an undirected (or directed) graph with vertex weights, which asks to find a minimum vertex cut with the smallest weight, where the weight of a vertex set is the sum of all the weights of the vertices in this set. This problem turns out to be much more challenging than the unweighted (or unit-weighted) case of vertex connectivity. The state-of-the-art result is by \cite{ChuzhoyT25}, which gives a $\min(mn^{0.99+o(1)},m^{1.5+o(1)})$ algorithm for minimum vertex cut on undirected graphs where $m$ denotes the number of edges. Although it is plausible that their techniques can be extended to directed graphs, this was not discussed in their paper, leaving the $\tO{mn}$\footnote{We use $\tO{\cdot }$ to hide polylogarithmic factors.} time algorithm of \cite{HenzingerRG00} as the best known result for directed graphs.

Our reduction \Cref{thm:intro_1} also holds on weighted graphs. Strictly speaking, we have the following result. 

\begin{theorem}[Informal version of \Cref{thm:reduction}]\label{thm:intro_2}
   The directed weighted minimum vertex cut problem on an $n$-vertex directed graph reduces to the undirected weighted minimum vertex cut problem on a graph with $2n$ vertices and $O(n^2)$ edges.
\end{theorem}

By plugging in \Cref{thm:intro_2} with the $O(\min(mn^{0.99+o(1)},m^{1.5+o(1)}))$ undirected minimum vertex cut algorithm \cite{ChuzhoyT25}, we immediately get the following corollary that improves the $\tO{mn}$ algorithm \cite{HenzingerRG00} from 30 years ago on dense graphs.

\begin{corollary}\label{cor:intro_dirVC}
    There is an $n^{2.99+o(1)}$ time randomized algorithm for directed minimum vertex cut.
\end{corollary}

\paragraph{Remark.} Our construction extends to Steiner vertex cut problems as well, but this does not immediately simplify the construction of \cite{HuangLSW23} because there was no known conditional lower bound for the directed Steiner vertex cut problem before their work.

\paragraph{Beyond the sequential model.} We believe the simplicity of our reduction will make it easy to extend to other computation models, including communication, distributed, or streaming models~\cite{BlikstadJMY25}. These models are especially suitable for our reduction theorem because many of them use $n$ as the parameter. However, individual implementation involves delicate technical details of the reduction that are beyond the scope of this paper.  %

\paragraph{Comparison to Edge Connectivity.} 
Edge connectivity asks for the minimum number of edges whose removal disconnects an undirected graph or breaks the strong connectivity of a directed graph. Although this problem appears very similar to vertex connectivity, there is \emph{no} known reduction between undirected and directed edge connectivity, making our result even more surprising.
In fact, directed edge connectivity is generally considered harder to handle than its undirected counterpart. It has long been known that undirected edge connectivity can be solved in nearly linear time using (relatively) simple combinatorial techniques~\cite{Karger00}. In contrast, directed edge connectivity implies solving \emph{unit-capacity maximum flow}, which was only recently shown to be solvable in almost linear time even on dense graphs~\cite{BrandLLSS0W21,ChenKLPGS25}, using sophisticated continuous optimization methods combined with complex data structures.

\paragraph{Independent Work.}  Independently of our work, the running time of the new weighted vertex connectivity algorithms have been improved to  $O\left(\min\{m n^{11/12+o(1)}\cdot d^{1/12}, n^{2.677}\}\cdot (\log W)^{O(1)}\right ) \leq O\left(m n^{0.976}\cdot (\log W)^{O(1)}\right)$, where $d$ is the average vertex degree in of the input directed graph and $W$ is the ratio between the largest and smallest vertex weight~\cite{ChuzhoyMosenzonTrabelsi2026}. They design algorithms that work on directed graphs, whereas our results are based on a black-box reduction to undirected graphs. Furthermore,  they also show an algorithm for unweighted directed graphs with running time of $ O(\min\{m^{1+o(1)} k,n^{2+o(1)}\})$ where $k$ is the vertex connectivity.

\section{Preliminaries}
\newcommand{\NoG}{N^\text{out}_G}
\newcommand{\NiG}{N^\text{in}_G} 

In this paper, we consider both undirected and directed graphs $G=(V,E)$. 
For undirected graphs, we write the edge set $E \subseteq \binom{V}{2}$ as a set of two-element subsets of $V$ and we use $N_G(u)\subseteq V$ to denote the neighborhood set of $u \in V$ in $G$. 
For directed graphs, we write the edge set $E \subseteq V \times V$ as a set of ordered pairs of $V$ and we use $\NoG(u),\NiG(u)$ to denote the out- and in-neighbors of $u \in V$ in $G$. 
When $G$ is vertex-weighted, each vertex is assigned a weight $w_G(v)\in \bbZ^+$ (when $G$ is clear from the context, we omit $G$ and write $w(v)$). If each vertex is assigned weight $1$, then the weight of a vertex set is simply its cardinality.

\paragraph{Vertex cut.} A \emph{vertex cut} $(L,S,R)$ is a tri-partition of the vertex set $V$ such that there are no edges from $L$ to $R$ and both $L$ and $R$ are non-empty sets. The vertex set $S$ is called a \emph{separator} (following the literature, we also refer to it as a \emph{vertex cut}). The size of the vertex cut is $|S|$ and the weight of the vertex cut is $w(S):=\sum_{v\in S}w(v)$. If $s\in L$ and $t\in R$, then we call $(L,S,R)$ (or $S$) a $(s,t)$-vertex cut. It is an $s$-vertex cut if $s \in L$.  Throughout this paper, we assume that every input graph is non-complete and thus a vertex cut always exists. 
\paragraph{Minimum vertex cut.} The problem of \emph{minimum vertex cut} asks to find a vertex cut $S$ that minimizes the weight $w(S)$ (in the case of a weighted graph) or the size $|S|$ (in the case of an unweighted graph). The latter case is also referred to as \emph{vertex connectivity}.
\paragraph{Single-source minimum vertex cut.} The problem of \emph{single source minimum vertex cut} asks, given a graph and a source $s\in V$, to find an $s$-vertex cut $(L,S,R)$ that minimizes $w(S)$ (in the case of a weighted graph) or the size $|S|$ (in the case of an unweighted graph). We denote $\kappa_G(s) = w(S)$ or $|S|$ in the case of an unweighted graph.
\paragraph{All-pairs minimum vertex cut.} The problem of \emph{all-pairs minimum vertex cut} asks, given a graph $G$, to find a minimum vertex cut $(L,S,R)$ such that $s\in L, t\in R$, for every pair of nodes $(s,t)\in V$ (or report that such cut does not exist). We denote $\kappa_G(s,t) = w(S)$ or $|S|$ in the case of an unweighted graph. 
\paragraph{Steiner minimum vertex cut.} The problem of \emph{Steiner minimum vertex cut} asks, given a graph $G$ and a terminal set $T \subseteq V$, to find a minimum vertex cut $(L,S,R)$ such that $T \cap L \neq \emptyset$ and $T \cap R \neq \emptyset$, that is, a minimizer of $\kappa_G(s,t)$ over all pairs $s,t \in T$. 

\section{A Reduction from Directed to Undirected Vertex Connectivity}

We first state our main result and then describe the reduction. 

\begin{theorem} \label{thm:reduction}
    Given a vertex-weighted directed graph $G=(V,E,w)$ with $|V|=n$, there is a vertex-weighted undirected graph $G'=(V',E',w')$ with $2n$ vertices and $O(n^2)$ edges such that the following directed vertex cut variants in $G'$ can be solved using an algorithm for the undirected problem in $G$: global, all-pairs, single-source/sink, Steiner. For the Steiner variant, the terminal set $T'$ given as input to the undirected problem has size $2|T|$, where $T$ is the terminal set given as input to the directed problem.  
\end{theorem}

\paragraph{Construction. }
Consider a vertex-weighted directed graph $G = (V, E, w)$. We define the vertex-weighted undirected graph $G' = (V', E', w')$ as follows. We create two copies $V_\text{out}$ and $V_\text{in}$ of $V$ and connect all vertices within each copy, creating two cliques of size $n$. We then add a perfect matching connecting each vertex with its copy. Finally, we add the directed edges in $G$ from copy $V_\tout$ to copy $V_\tin$. Formally, we have
\begin{itemize}
    \item $V' = V_\text{out} \cup V_\text{in}$, where $V_\text{out} = \{v_\text{out} : v \in V\}$ and $V_\text{in} = \{v_\text{in} : v \in V\}$ are two copies of $V$;
    \item $E' = E_\text{out} \cup E_\text{in} \cup E_\text{matching} \cup E_{\text{out} \rightarrow \text{in}} $, where 
    \begin{itemize}
        \item $E_\text{out} = \binom{V_\text{out}}{2}$ and $E_\text{in} = \binom{V_\text{in}}{2}$ are the cliques over $V_\tout$ and $V_\tin$, respectively,
        \item $E_\text{matching} = \{ \{v_\tout, v_\tin\} : v \in V\}$ is the perfect matching between $V_\tout$ and $V_\tin$, and
        \item $E_{\text{out} \rightarrow \text{in}} = \{ \{ u_\tout, v_\tin \} : (u, v) \in E \}$ are the directed edges from $G$ embedded in $G'$;
    \end{itemize}
    \item $w'(v_\tout) = w'(v_\tin) = w(v)$ for all $v \in V$.
\end{itemize}

This construction should be contrasted with the standard \emph{split graph construction}~\cite{DantzigFulkerson1956} where we split every vertex $v \in V$ to $v_{\text{out}}$ and $v_{\text{in}}$ and add appropriate edges, which was known since 1956; this construction has been used in many places including  max-flow based algorithms (e.g., see~\cite{HenzingerRG00,Gabow06,NanongkaiSY19,ForsterNYSY20,LNPSY25}). In other words, the graph $G'$ is the split graph $SG$ after removing directions and adding cliques on $V_{\text{out}}$ and another one on $V_{\text{in}}$, see \Cref{fig:example} for illustration.

\begin{figure}[h] %
    \centering
    \includegraphics[width=0.8\textwidth]{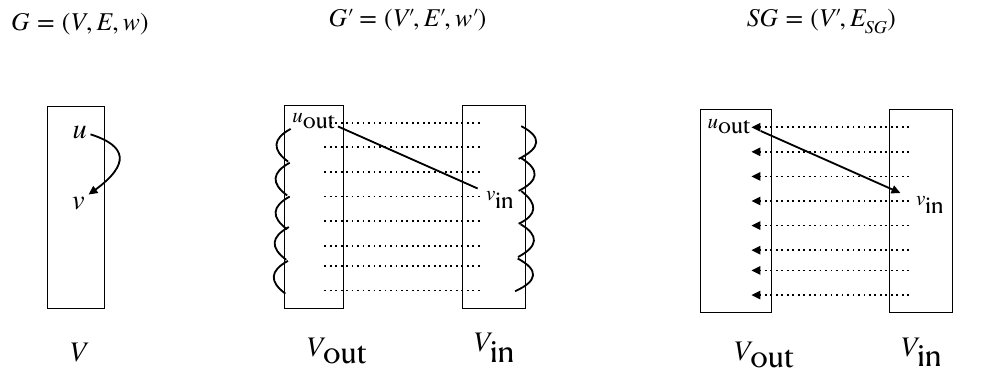} %
    \caption{ (Left) An input graph $G = (V,E,w)$. (Middle) The output graph $G' = (V',E',w')$. (Right) The split graph construction (ignoring weights) $SG$ given $G$.  The split graph usually has the edge weight as capacity in one direction; our graph has vertex weights and no edge weights. }
    \label{fig:example}
\end{figure}

\subsection{Basic Facts}

We first make two simple observations about vertex cuts in $G'$. 
\begin{proposition}\label{prop:cuts-in-G'}
For every vertex cut $(L', S', R')$ in $G'$, we either have $L' \subseteq V_\text{out}, R' \subseteq V_\text{in}$ or $L' \subseteq V_\text{in}, R' \subseteq V_\text{out}$.
\end{proposition}
\begin{proof}
    There are no edges in $G'$ between $L'$ and $R'$, so each clique must be a subset of $L' \cup S'$ or $R' \cup S'$. As $V_\text{out}$ and $V_\text{in}$ are cliques and form a partition of the vertex set $V'$, the claim follows. 
\end{proof}

\begin{proposition}\label{prop:cuts-wlog}
Without loss of generality, we can write every vertex cut in $G'$ as $(L', S', R')$ such that $L' \subseteq V_\text{out}$ and $R' \subseteq V_\text{in}$. Moreover, if we put $L = \{ v \in V : v_\tout \in L' \}$ and $R = \{ v \in V : v_\tin \in R' \}$, then the sets $L$ and $R$ are nonempty and disjoint, and there is no edge $(u,v) \in E$ such that $u \in L$ and $v \in R$. 
\end{proposition}
\begin{proof}
The first part immediately follows from Proposition \ref{prop:cuts-in-G'} and the fact that $G'$ is undirected, so $L'$ and $R'$ can be swapped. To see the second part, assume for contradiction that there is some vertex $v \in L \cap R$, that is, $v_\tout \in L'$ and $v_\tin \in R'$. Since we have the perfect matching edge $\{v_\tout, v_\tin\} \in E_{\text{id}}$ in $G'$, this contradicts the fact that $(L', S', R')$ is a cut. Now assume for contradiction that there is some edge $(u,v) \in E$ such that $u \in L$ and $v \in R$. Again, this contradicts the fact that $(L', S', R')$ is a cut as there is the edge $\{u_\tout, v_\tin\} \in E_{\tout \rightarrow \tin}$ in $G'$.
\end{proof}

Next, we show the key property of our construction. 
\begin{lemma}\label{lemma:neighbors-G'-G}
Let $\emptyset \neq R_\tin \subseteq V_\text{in}$ and let $R = \{ v \in V : v_\tin \in R_\tin \}$. We have $w'(N_{G'}(R_\tin)) = w(N_G^\tin (R)) + w(V)$.
\end{lemma} 

\begin{proof}
We proceed by splitting the neighborhood of $R_\tin$ into its parts within $V_\tin$, $R_\tout$ and $V_\tout \setminus R_\tout$, where $R_\tout = \{ v_\tout : v_\tin \in R_\tin \}$. 
First, we consider $N_{G'}(R_\tin) \cap V_\tin$. As $V_\tin$ is a clique, we have $N_{G'}(R_\tin) \cap V_\tin = V_\tin \setminus R_\tin$ and hence $w'(N_{G'}(R_\tin) \cap V_\tin) = w'(V_\tin \setminus R_\tin) = w(V \setminus R)$. 
Second, we consider $N_{G'}(R_\tin) \cap R_\tout$. As the perfect matching edges $E_\text{matching}$ connect each $v_\tin \in R_\tin$ with its counterpart $v_\tout$, this set is exactly $R_\tout$, hence $w'(N_{G'}(R_\tin) \cap R_\tout) = w'(R_\tout) = w(R)$. 
Third, we consider $N_{G'}(R_\tin) \cap (V_\tout \setminus R_\tout)$. These are neighours of $R_\tin$ from those edges in $E_{\tout \rightarrow \tin}$ that do not have an endpoint in $R_\tout$. More precisely, we have 
\begin{align*}
N_{G'}(R_\tin) \cap (V_\tout \setminus R_\tout) 
&= \{ u_\tout : \{u_\tout, v_\tin\} \in E_{\tout \rightarrow \tin}, \ u_\tout \in V_\tout \setminus R_\tout, \ v_\tin \in R_\tin \} 
\\&= \{ u_\tout : (u,v) \in E, \ u \in V \setminus R, \ v \in R \} 
\\&= N_G^\tin (R).    
\end{align*}

It follows that $w'(N_{G'}(R_\tin) \cap (V_\tout \setminus R_\tout)) = w(N_G^\tin (R))$. Putting the three parts together, we conclude that 
\[
w'(N_{G'}(R_\tin)) = w(V \setminus R) + w(R) + w(N_G^\tin (R)) = w(N_G^\tin (R)) + w(V).\qedhere
\]
\end{proof}

\subsection{Global minimum vertex cut} The following lemma establishes a precise relationship between the minimum vertex cuts of $G'$ and $G$. 
\begin{lemma}\label{lemma:min-cut-G'-G}
$\kappa_{G'} = \kappa_G + w(V)$.    
\end{lemma}
\begin{proof}
Let $(L', S', R')$ be a minimum vertex cut in $G'$. By Proposition \ref{prop:cuts-wlog}, we may assume without loss of generality that $L' \subseteq V_\tout$ and $R' \subseteq V_\tin$. Put $L = \{ v \in V : v_\tout \in L' \}$ and $R = \{ v \in V : v_\tin \in R' \}$; then $L$ and $R$ are nonempty and disjoint, and there is no edge $(u,v) \in E$ such that $u \in L$ and $v \in R$. It follows that $w'(N_{G'}(R')) = w(N_G^\tin (R)) + w(V)$ by Lemma \ref{lemma:neighbors-G'-G}. Since $(L', S', R')$ is a minimum vertex cut and all vertex weights are positive, we have $S' = N_{G'}(R')$, so the weight of the cut is $w'(S') = w(N_G^\tin (R)) + w(V)$. Therefore, $(L \cup X, N_G^\tin(R), R)$, where $X = V \setminus (L \cup R \cup N_G^\tin(R))$, is a directed vertex cut in $G$ with weight $w(N_G^\tin(R))$. This concludes the first direction $\kappa_{G'} \geq \kappa_G + w(V)$. 

To prove the second direction, let $(L, S, R)$ be a minimum directed vertex cut in $G$. Let $L_\tout = \{ v_\tout : v \in L \}$, $R_\tin = \{ v_{\tin} : v \in R \}$. As $(L,S,R)$ is a directed vertex cut in $G$, the sets $L$ and $R$ are nonempty and disjoint, and there are no edges $(u,v) \in E$ such that $u \in L$ and $v \in R$. It follows from the definition of $G'$ that there are no edges between $L_\tout$ and $R_\tin$, and the two sets are nonempty and disjoint. Let $X = V_\tout \setminus (L_\tout \cup N_{G'}(R_\tin))$ and observe that there are no edges between $X$ and $R_\tin$. 
Then $(L_\tout \cup X, N_{G'}(R_\tin), R_\tin)$ is a vertex cut in $G'$. It has weight $w(N_{G'}(R_\tin)) = w(N_G^\tin (R)) + w(V)$ by Lemma \ref{lemma:neighbors-G'-G}, which gives us an upper bound on the minimum vertex cut size $\kappa_{G'}$ in $G'$. By the minimality of $(L,S,R)$, we have $S = N_G^\tin(R)$, so $\kappa_G = w(N_G^\tin(R))$ and hence $\kappa_{G'} \leq \kappa_G + w(V)$. 
\end{proof}

\subsection{All-pairs minimum vertex cut} For the all-pairs case, we proceed analogously to the global case. We use the following notation. 
For distinct $s, t \in V$, let $\kappa_G(s,t)$ denote the weight of the minimum weighted directed $s$-$t$-cut in $G$. 
For distinct $s', t' \in V'$, let $\kappa_{G'}(s',t')$ denote the weight of the minimum weighted undirected $s'$-$t'$-cut in $G'$. 
Note that these values are only defined if $(s, t) \notin E$ and $\{s', t'\} \notin E'$, respectively; otherwise we can treat them as $\infty$. 
Then $\kappa_{G}(s, t)$ is embedded in $G'$ as follows.

\begin{lemma}\label{lemma:pair-min-cut-G'-G}
Let $s,t \in V$ such that $s \neq t$ and $(s,t) \notin E$. Then $\kappa_{G}(s, t) = \kappa_{G'}(s_\tout, t_\tin) - w(V)$.
\end{lemma}
\begin{proof}
Given that $s \neq t$ and there is no edge $(s,t) \in E$, there is also no edge $\{s_\tout, t_\tin\} \in E'$, hence there exists as minimum vertex cut in $G'$ separating $s_\tout$ from $t_\tin$.

Let $(L', S', R')$ be a minimum vertex cut in $G'$ separating $s_\tout$ from $t_\tin$. Following Proposition \ref{prop:cuts-wlog}, we may assume without loss of generality that $s_\tout \in L' \subseteq V_\tout$ and $t_\tin \in R' \subseteq V_\tin$. Put $L = \{ v \in V : v_\tout \in L' \}$ and $R = \{ v \in V : v_\tin \in R' \}$; then $L$ and $R$ are disjoint, $s \in L$, $t \in R$, and there is no edge $(u,v) \in E$ such that $u \in L$ and $v \in R$. It follows that $w'(N_{G'}(R')) = w(N_G^\tin (R)) + w(V)$ by Lemma \ref{lemma:neighbors-G'-G}. By the minimality of $(L', S', R')$ and as all vertex weights are positive, we have $S' = N_{G'}(R')$, so the weight of the cut is $w'(S') = w(N_G^\tin (R)) + w(V)$. Therefore, $(L \cup X, N_G^\tin(R), R)$, where $X = V \setminus (L \cup R \cup N_G^\tin(R))$, is a directed vertex cut separating $s$ and $t$ in $G$ with weight $w(N_G^\tin(R))$. This concludes the first direction $\kappa_{G'}(s_\tout,t_\tin) \geq \kappa_G(s,t) + w(V)$

To prove the second direction, let $(L, S, R)$ be a minimum directed vertex cut separating $s \in L$ from $t \in R$ in $G$. Let $L_\tout = \{ v_\tout : v \in L \}$, $R_\tin = \{ v_{\tin} : v \in R \}$. As $(L,S,R)$ is a directed vertex cut in $G$, the sets $L$ and $R$ are nonempty and disjoint, and there are no edges $(u,v) \in E$ such that $u \in L$ and $v \in R$. It follows from the definition of $G'$ that $s_\tout \in L_\tout$, $t_\tin \in R_\tin$, there are no edges between $L_\tout$ and $R_\tin$, and the two sets are disjoint. Let $X = V_\tout \setminus (L_\tout \cup N_{G'}(R_\tin))$ and observe that there are no edges between $X$ and $R_\tin$ in $G'$. 
Then $(L_\tout \cup X, N_{G'}(R_\tin), R_\tin)$ is a vertex cut in $G'$ separating $s_\tout \in L_\tout$ from $t_\tin \in R_\tin$. It has weight $w(N_{G'}(R_\tin)) = w(N_G^\tin (R)) + w(V)$ by Lemma \ref{lemma:neighbors-G'-G}, which gives us an upper bound on $\kappa_{G'}(s_\tout, t_\tin)$ in $G'$. By the minimality of $(L,S,R)$, we have $S = N_G^\tin(R)$, so $\kappa_G(s,t) = w(N_G^\tin(R))$ and hence $\kappa_{G'}(s_\tout,t_\tin) \leq \kappa_G(s,t) + w(V)$.
\end{proof}

\subsection{Single-source/sink minimum vertex cut } Recall that, in addition to the input graph $G=(V,E)$, we are given a source vertex $s$ or a sink vertex $t$. We denote the single-source and single-sink connectivity by $\kappa_{G}(s, \ast)$ and $\kappa_{G}(\ast, t)$, respectively. As single-source and single-sink cuts in the undirected case are equivalent, we simply write $\kappa_{G'}(v)$ for $v \in V'$.

\begin{corollary}\label{corollary:single-source-min-cut-G'-G}
$\kappa_{G}(s, \ast) = \kappa_{G'}(s_\tout) - w(V)$ and $\kappa_{G}(\ast, t) = \kappa_{G'}(t_\tin) - w(V)$.
\end{corollary}
\begin{proof}
Using the fact that $V_\tout$ and $V_\tin$ are cliques and Lemma \ref{lemma:pair-min-cut-G'-G}, we have 
\begin{align*}  
\kappa_{G'}(s_\tout) &= \min_{v \in V} \kappa_{G'}(s_\tout, v_\tin) = \min_{v \in V} \kappa_{G}(s, v) + w(V) = \kappa_{G}(s, \ast) + w(V) \ \text{ and } \\ 
\kappa_{G'}(t_\tin) &= \min_{v \in V} \kappa_{G'}(v_\tout, t_\tin) = \min_{v \in V} \kappa_{G}(v, t) + w(V) = \kappa_{G}(\ast, t) + w(V) . \qedhere\end{align*}
\end{proof}

\subsection{Steiner minimum vertex cut } Recall that, in addition to the input graph $G=(V,E)$, we are given a terminal set $T$. Let $T' = T_\tout \cup T_\tin$. We show that the undirected Steiner minimum vertex cut solution in $G'$ with terminal set $T'$ corresponds to a directed optimal solution. 

\begin{corollary}\label{corollary:steiner-min-cut-G'-G}
$\min_{s,t \in T } \kappa_{G}(s, t) = \min_{s', t' \in T} \kappa_{G'}(s', t') - w(V)$.
\end{corollary}
\begin{proof}
Let $s', t' \in T'$ be arbitrary. If $s', t' \in V_\tin$ or $s', t' \in V_\tout$, there is a clique edge between $s'$ and $t'$, so $\kappa_{G'}(s', t') = \infty$. Therefore, we may without loss of generality (as $G'$ is undirected) assume that $s' \in V_\tout \cap T' = T_\tout$ and $t' \in V_\tin \cap T' = T_\tin$. Then $s' = s_\tout$ and $t' = t_\tin$ for some $s,t \in T$. We may assume $s \neq t$ and $(s,t) \notin E$, otherwise we again have $\kappa_G(s,t) = \kappa_{G'}(s', t') = \infty$. Then by Lemma \ref{lemma:pair-min-cut-G'-G}, we have $\kappa_{G}(s, t) = \kappa_{G'}(s', t') - w(V)$, and the claim follows. 
\end{proof}

\subsection{Combining Together}
We are ready to prove Theorem \ref{thm:reduction}.
\begin{proof}[Proof of Theorem \ref{thm:reduction}]
By construction, $G'$ has $2n$ vertices and $2\binom{n}{2} + n + m \leq O(n^2)$ edges and the terminal set $T'$ for the Steiner variant has size $2|T|$. 
By Lemmas \ref{lemma:min-cut-G'-G}, \ref{lemma:pair-min-cut-G'-G} and Corollaries \ref{corollary:single-source-min-cut-G'-G}, \ref{corollary:steiner-min-cut-G'-G}, the solutions in $G'$ for global, all-pairs, single-source/sink, and Steiner variants of vertex cuts correspond to directed solutions in $G$. Moreover, given an undirected vertex cut $(L_\tout, S', R_\tin)$ in $G'$, we can efficiently extract a directed vertex cut $(L, S, R)$ in $G$, where $S = V \setminus (L \cup R)$ and $w(S) = w(S') -w(V)$, 
\end{proof}

\section{Concluding Remarks} One can notice that our reduction from directed to undirected graph has a loss in boosting the number of edges to $O(n^2)$, which is not a good reduction when the graph is moderately sparse. It is a great open problem to prove that directed vertex connectivity can be solved in almost linear time for any sparsity, not necessarily done by a black-box reduction.

\subsection*{Acknowledgment} 
This research was partially supported by the ETH Zürich Foundation, Dr. Max Rössler, and the Walter Haefner Foundation.

\bibliographystyle{alpha}
\bibliography{refs,refs-sorrahcai}

\end{document}